\documentclass[conference,10pt]{IEEEtran}
\usepackage[T1]{fontenc}

\usepackage{amssymb,amsmath,amsfonts,amsthm}
\usepackage{mathrsfs}
\usepackage{cite}
\usepackage{array}
\usepackage{tabularx}
\usepackage[table]{xcolor}
\usepackage{multicol, multirow}
\usepackage{color}
\usepackage{mathtools}
\usepackage{subcaption}
\usepackage{mathtools}
\usepackage{tikz,pgf}
\usetikzlibrary{calc,positioning,mindmap,trees,decorations.pathreplacing}
\usepackage{graphicx}
\usepackage{bbm}
\usepackage[utf8]{inputenc}
\usepackage{algorithm}
\usepackage{algpseudocode}

\newtheorem{definition}{Definition}

\newtheorem{lemma}{Lemma}

\newtheorem{assumption}{Assumption}
\newcommand{\argmax}{\mathop{\rm arg~max}\limits}

\usepackage[colorinlistoftodos,prependcaption,backgroundcolor=black!5!white,bordercolor=red]{todonotes}

\IEEEoverridecommandlockouts
\ifCLASSINFOpdf
\else
\fi
\usepackage[cmintegrals]{newtxmath}
\begin{document}
	\title{Decentralized Learning over Wireless Networks: The Effect of Broadcast with Random Access}
	
	\author{Zheng Chen, Martin Dahl, and Erik G. Larsson\\
		Department of Electrical Engineering, Link\"{o}ping University, Sweden.\\ Email: \{zheng.chen@liu.se, marda545@student.liu.se,  erik.g.larsson@liu.se\}
		\thanks{This work was supported in part by Zenith, Excellence Center at Link\"{o}ping - Lund in Information Technology (ELLIIT), Swedish Research Council, and Knut and Alice Wallenberg (KAW) Foundation.}}
	
	\maketitle
	
	\begin{abstract}	
		In this work, we focus on the communication aspect of decentralized learning, which involves multiple agents training a shared machine learning model using decentralized stochastic gradient descent (D-SGD) over distributed data. In particular, we investigate the impact of broadcast transmission and probabilistic random access policy on the convergence performance of D-SGD, considering the broadcast nature of wireless channels and the link dynamics in the communication topology. Our results demonstrate that optimizing the access probability to maximize the expected number of successful links is a highly effective strategy for accelerating the system convergence.
	\end{abstract} 
	
	
	\section{Introduction}
	Decentralized learning is rooted in the theoretical framework of multi-agent optimization, where a group of agents collaborative in minimizing a common objective function \cite{decentrazlied-first-order}. Many methods and algorithms have been developing for solving such type of distributed optimization problems, such as distributed sub-gradient \cite{distributed-subgradient} and distributed ADMM \cite{distributed_admmd}. This paper focuses on the decentralized stochastic gradient decent (D-SGD) method, in which each agent combines local gradient computation with consensus-based model updating in an iterative manner \cite{yuan2016convergence, lian2017can, dso-ml}.

	A crucial aspect of D-SGD (and its variants) is the consensus formation among the agents, which heavily relies on information exchange and fusion within the network \cite{consensus-cooperation}. Although this linear averaging step appears straightforward from a mathematical perspective, the coordination of information exchange (transmission and reception) among agents in wireless networks is a non-trivial task. 
	Specifically, within the consensus updating step, each agent sends the same information (local model parameters) to its neighbors, which can be done through a single broadcast transmission rather than multiple link-based transmissions. In the meanwhile, concurrent transmissions from multiple nodes will create interference among them, which can lead to failed reception of information at the receiver side. Several existing theoretical studies have considered the effect of unreliable communication in decentralized federated learning over wireless networks. \cite{ye2021decentralized, fl-d2d, async-dl}. However, it is still unclear how to design appropriate medium access control (MAC) protocols and interference management schemes for achieving consensus-oriented communication in wireless networks. Further investigation is required to fully understand the impact of communication design on link dynamics, and ultimately on the performance of D-SGD over wireless networks.

	Recent research has explored ways to customize the communication pattern in decentralized learning to optimize convergence speed and reduce communication costs \cite{MATCHA, laplacian_sampling}. These studies focus on link-based scheduling instead of broadcast-based scheduling. 
	Incorporating broadcast transmission of information introduces additional challenges in the communication scheduling strategy, since the weight matrix in every iteration cannot be guaranteed to be symmetric due to the asymmetric information flow in broadcast communication.

	In this work, we focus on MAC layer communication scheme for decentralized learning with probabilistic random access and broadcast transmission. In every iteration, nodes access the channel and broadcast their model updates with a certain probability. A node can successfully receive a packet if there is only one neighbor broadcasting. Collision occurs when multiple nodes broadcast to a common neighbor. Based on this simple success or collision model, we demonstrate that there exists a strong correlation between the access probability that maximizes the number of successful links and the probability that maximizes the second-largest eigenvalue of the expected weight matrix. This finding provides valuable insights into the design of random access protocols with spatial reuse of resources, for the purpose of accelerating convergence in decentralized learning over large-scale wireless networks.

	\section{System Model}
	We consider a decentralized learning system where $N$ nodes collaborate in training a shared machine learning (ML) model, parameterized by a vector $\boldsymbol{x}\in \mathbb{R}^d$. The goal of model training is to find the optimal model parameter vector as the solution to the following problem
	\vspace{-0.1cm} 
	\begin{equation}
		\min\limits_{\boldsymbol{x}\in \mathbb{R}^d} F(\boldsymbol{x})=\frac{1}{N}\sum_{i=1}^{N}F_i(\boldsymbol{x}),
		\label{eq:prob-def}
	\end{equation}
	where $F(\boldsymbol{x})$ is the global objective function and $F_i(\boldsymbol{x})$ is the local objective function at node $i$. Let $ \mathcal{D}_i$ represent the training data available at the $i$-th node, then the local objective function can be written as 
	\vspace{-0.1cm} 
	\begin{equation}
		F_i(\boldsymbol{x})=\frac{1}{|\mathcal{D}_i|}\sum_{s\in \mathcal{D}_i}l(\boldsymbol{x},s),
	\end{equation}
	where $l(\boldsymbol{x},s)$ is the local loss function for sample $s$.

	\subsection{Graph Model for Network Connectivity}
	The network is modeled as an undirected graph $\mathcal{G}(\mathcal{V},\mathcal{E})$ with $\mathcal{V}=\{v_1,\ldots, v_N\}$ representing the set of nodes and $\mathcal{E}\subseteq\mathcal{V}\times\mathcal{V}$ representing the set of links. The connectivity of the graph is described by the adjacency matrix $\mathbf{A}\in\mathbb{R}^{N\times N}$, where $A_{ij}$ (the $i$-th row and $j$-th column of $\mathbf{A}$) is  $1$ if $(i,j)\in\mathcal{E}$, and $0$ otherwise. Let $\mathcal{N}_i=\{j\in\mathcal{V} |(i,j)\in\mathcal{E}\}$ denote the set of neighbors of node $i$. The degree of node $i$ is defined as $d_i=\sum_{j=1}^{N}A_{ij}=|\mathcal{N}_i|$. The degree matrix is defined as $\mathbf{D}=\textrm{diag}(d_1,\ldots, d_N)$. The Laplacian matrix is $\textbf{L}=\textbf{D}-\textbf{A}$.
	
	Theoretically, all nodes in a wireless environment are ``connected'' due to the broadcast nature of wireless channels. In this work, to simplify our analysis, we consider that any pair of nodes $(i,j)$ has a well-defined connectivity indicator with binary status, e.g., two nodes are considered to be connected when the link distance is smaller than a threshold. 
	
	\subsection{D-SGD over Networked Agents}
	Consensus-based D-SGD is a commonly used algorithm for solving the decentralized optimization problem defined in \eqref{eq:prob-def}. The plain version of D-SGD consists of three main steps: 1) local stochastic gradient computation; 2) communication with neighbors; 3) consensus-based model fusion and updating.
	
	Let $g_i(\boldsymbol{x}_i^{(t)})=\nabla F_i(\boldsymbol{x}_i^{(t)};s_i)$ denote the stochastic gradient vector at node $i$ computed over one or a subset of randomly selected data samples $s_i\in \mathcal{D}_i$ in iteration $t$. The model parameter vector updates by the following iteration rule
	\begin{equation}\label{DSGD_equation}
		\boldsymbol{x}_i^{(t+1)}=\sum_{j=1}^{N}W_{j,i}^{(t)}\left[\boldsymbol{x}_j^{(t)}-\eta g_i(\boldsymbol{x}_i^{(t)})\right].
	\end{equation} 
	Here, $W_{i,j}^{(t)}$ indicates the weight that node $i$ assigns to the model update received from node $j$. We can write all weight coefficients in a matrix form $\boldsymbol{W}^{(t)}\in\mathbb{R}^{N\times N}$, referred to as the weight matrix or the mixing matrix.

	\begin{definition}
		A square and non-negative matrix $\mathbf{M}$ is called
		\begin{itemize}
			\item (row) stochastic if each row of the matrix sums to $1$;
			\item doubly stochastic if each row and each column sum to $1$. 
		\end{itemize} 
	\end{definition}
	
	Let $\overline{\boldsymbol{x}}^{(t)}=\frac{1}{N}\sum_{i=1}^{N}\boldsymbol x_i^{(t)}$ represent the average model in the current iteration $t$. As shown in \cite{lian2017can, MATCHA, laplacian_sampling}, the convergence of D-SGD (in the sense that $\frac{1}{T}\sum_{t=1}^{T}\mathbb{E}[\lVert \nabla F(\overline{\boldsymbol{x}}^{(t)})\rVert]$ becomes sufficiently small when $T$ increases) can be proved if the following assumptions hold.
	
	\begin{assumption}
		All local objective functions $F_i(x)$ are differentiable and the local gradients are L-Lipschitz continuous, i.e., $\lVert\nabla F_i (\boldsymbol{x_1})-\nabla F_i (\boldsymbol{x_2})\rVert\leq L \lVert\boldsymbol{x_1}-\boldsymbol{x_2} \rVert$, $\forall \boldsymbol{x_1}, \boldsymbol{x_2}\in \textnormal{dom~}F$.
	\end{assumption}
	\begin{assumption}
		The stochastic gradient at each node is an unbiased estimate of the true gradient of the local objective function, i.e., $\mathbb{E}[g_i(\boldsymbol{x})]=\nabla F_i (\boldsymbol{x})$.
	\end{assumption}
	
	\begin{assumption}
		The variance of the stochastic gradient at each node is uniformly bounded, i.e, $\mathbb{E}\left[\lVert g_i(\boldsymbol{x})-\nabla F_i (\boldsymbol{x})\rVert^2\right]\leq \sigma^2$.
	\end{assumption}
	
	\begin{assumption}
		The deviation between the local gradient at each node and the global gradient is bounded, i.e.,  $\mathbb{E}\left[\lVert \nabla F_i(\boldsymbol{x})-\nabla F (\boldsymbol{x})\rVert^2\right]\leq \xi^2$.
	\end{assumption}
	
	\begin{assumption}
		The mixing matrix $\boldsymbol{W}$ is symmetric and doubly stochastic, with the second largest absolute eigenvalue $\beta=\max\{|\lambda_2(\boldsymbol{W})|, |\lambda_N(\boldsymbol{W})|\}$ smaller than $1$.  
	\end{assumption}
	Note that $\beta$ is the spectral radius of $\mathbf{W}-\frac{1}{N}\textbf{1}\textbf{1}^{\transp}$, where $\textbf{1}$ is the all-ones column vector. From consensus perspective, the smaller $\beta$ is, the faster convergence we can achieve.

	\section{D-SGD with Random Access and Broadcast Transmission}
	We consider a probabilistic random access scheme for the broadcast transmission of model updates from all network nodes. In every iteration, the entire model parameter vector is considered as one packet, and its transmission consumes one time slot. Every node makes independent and random decisions (i.e., Bernoulli trials) on whether to access the channel and broadcast its current model or remain silent. We define $p$ as the access (or broadcast) probability of all nodes.

	\begin{definition}
		The \textit{broadcast} decision vector $\mathbf{b}^{(t)}\in\mathbb{R}^{N}$ is a vector whose $j$-th element is given as
		\begin{equation*}
			b_{j}^{(t)} = \begin{cases}
				1 &\text{if node j \textit{broadcasts} at iteration $t$}\\
				0 &\text{otherwise}
			\end{cases}.
		\end{equation*}
		With probabilistic random access policy, we have $\mathbb{E}[b_{j}^{(t)}]=p$.
	\end{definition}
	In every slot, a node can receive at most one packet successfully from its neighbors. When multiple nodes broadcast to a common neighbor, it will result in a collision and no information will be decoded, as illustrated by the example in Fig.~\ref{fig:example}. Then, we define a matrix $\mathbf{T}^{(t)}$ that contains binary variables indicating the status of each transmission.
	
	\begin{figure}[t!]
		\centering
		\includegraphics[width=0.7\columnwidth]{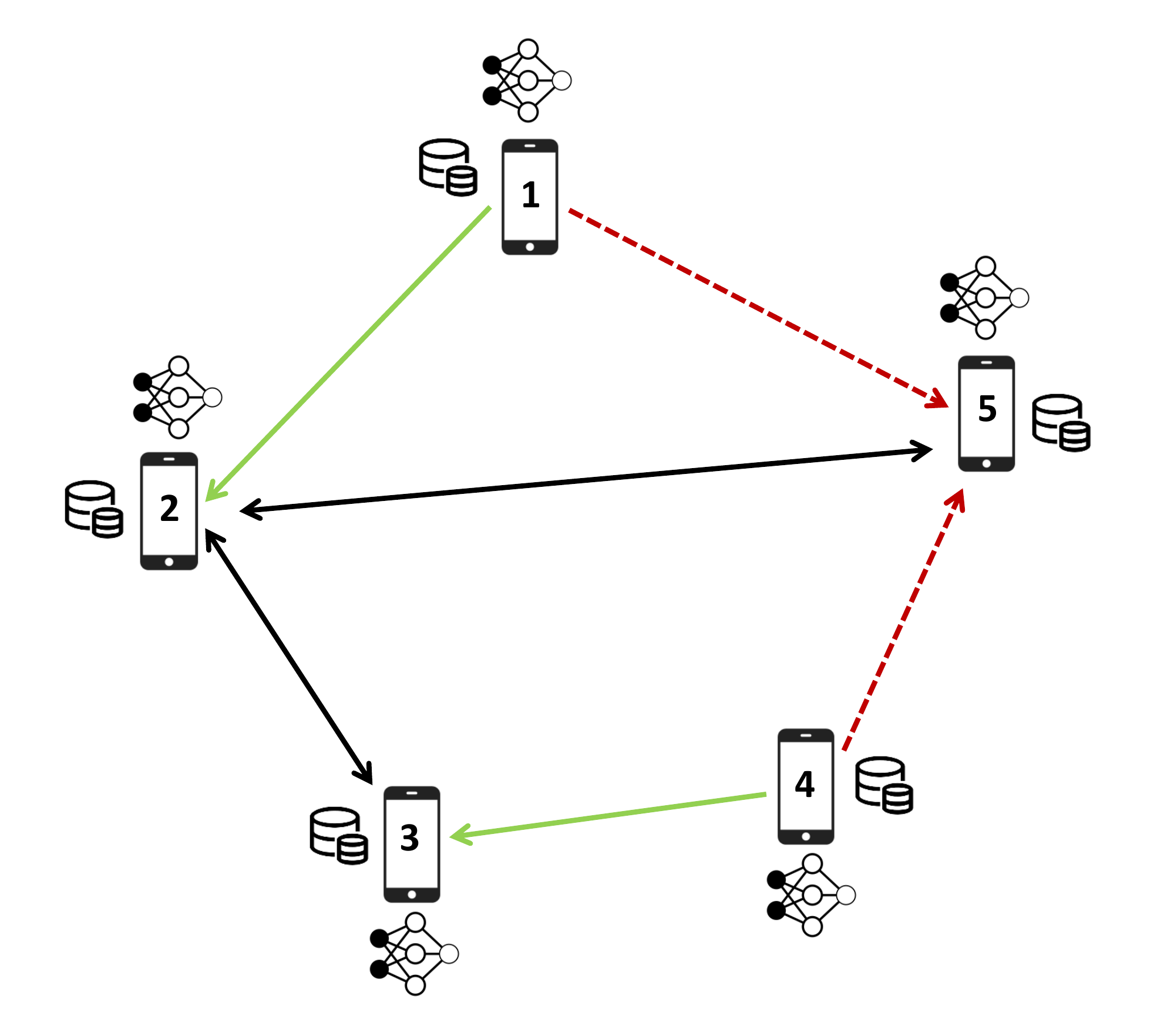}
		\caption{An example of random access with broadcast transmission. Node $1$ and $4$ broadcast simultaneously, causing collision at node $5$. The transmissions in the links $(1,2)$ and $(4,3)$ are successful.}
		\vspace{-0.1cm}
		\label{fig:example}
	\end{figure}

	\begin{definition}
		\label{def:trans_matrix}
		The \textit{transmission} status matrix $\mathbf{T}^{(t)}\in\mathbb{R}^{N\times N}$ is a square matrix where $(i,j)$-th element is a binary number indicating the outcome (success or failure) of the transmission of model update from node $j$ to node $i$, i.e.,
		\begin{equation*}
			T^{(t)}_{i,j} = \begin{cases}
				1 &\text{if node $j$ successfully transmits $\boldsymbol{x}_j^{(t)}$ to node $i$}\\
				0 &\text{otherwise.}
			\end{cases}.
		\end{equation*}
		For the diagonal element, we have $T^{(t)}_{i,i}=1, \forall i\in\{1,\ldots, N\}$. 
	\end{definition}

	Based on the success or collision assumption, a node $i$ can successfully receive information from its neighbor node $j$ if and only if: 1) node $j$ decides to broadcast; 2) node $i$ decides not to broadcast;
	3) all neighbors of node $i$ except node $j$ decide not to broadcast.
	Combining these conditions, we have
	\begin{equation}\label{transmission_dynamics}
		T^{(t)}_{i,j}=b^{(t)}_{j}(1-b_i^{(t)})\prod_{k\in\mathcal{N}(i)\setminus j}^{}(1-b_k^{(t)}), 
	\end{equation}
	for all $i\neq j$. Then, we obtain the probability of successful model transmission from node $j$ to node $i$ as	
	\begin{equation}
		p_{i,j}^{\text{suc}}=\mathbb{E}[T^{(k)}_{i,j}]=p(1-p)^{d_i},
		\label{eq:link-success}
	\end{equation}  
	where $d_i$ is the degree of node i.
	We refer to this probability as the link success probability of $(j,i)\in\mathcal{E}$.
	In general, $p_{i,j}^{\text{suc}}\neq p_{j,i}^{\text{suc}}$ due to the difference in their node degrees.

	\subsection{D-SGD with Link Failures}
	In the case with perfect communication and fixed topology, a common choice of the weight matrix design is 	
	\begin{equation}
		\mathbf{W} = \mathbf{I} - \epsilon \mathbf{L} \text{, where } \epsilon <\frac{1}{\max_{i\in\mathcal{V}}\{d_{i} \}}. 
		\label{eq:weight_initial}
	\end{equation}
	This choice of weight design connects the convergence speed of average consensus directly to the spectral property of the graph Laplacian.
	In our system model with random link failures caused by broadcast collision, using the initial weight design as in \eqref{eq:weight_initial},  we will obtain a time-varying weight matrix \vspace{-0.05cm} 
	\begin{equation}
		\mathbf{W}^{(t)} = \mathbf{W}\odot\mathbf{T}^{(t)},
	\end{equation}
	where $\odot$ is the Hadamard product and $\mathbf{T}^{(t)}$ is defined in Definition~\ref{def:trans_matrix}. Note that $\mathbf{W}^{(t)}$ is not guaranteed to be row- or column-stochastic. To compensate for the missing information, we apply the biased compensation method in \cite{fagnani2009average}, which allows every node to add the weights of the failed links to its own previous estimate. This strategy will produce a new row-stochastic  matrix $\overline{\mathbf{W}}^{(t)}$ with the $(i,j)$-th element being 
	
	\begin{equation}\label{eq:w-compensated}
		\overline{W}^{(t)}_{i,j} = \begin{cases}
			W^{(t)}_{i,j}T^{(t)}_{i,j} &\text{if $i\neq j$}\\
			1 - \sum_{k=1, k\neq i}^{n}W^{(t)}_{i,k}T^{(t)}_{i,k}  &\text{if $i=j$.}
		\end{cases},
	\end{equation} 
	Eventually, in every iteration, the model parameter vector at node $i$ updates by the following rule 	
	\begin{equation}
		\boldsymbol{x}_i^{(t+1)}=\sum_{j=1}^{N}\overline{W}_{i,j}^{(t)}\left[\boldsymbol{x}_j^{(t)}-\eta g(\boldsymbol{x}_i^{(t)})\right].
	\end{equation} 
	
	Even though the network  topology is originally modeled as an undirected graph, this asymmetric link success/failure will cause the new weight matrix $\overline{\mathbf{W}}^{(t)}$ to be non-symmetric. Naturally, this implies that some extra bias will be introduced in the converged model.

	\begin{algorithm}[t!]
		\caption{D-SGD with random access and broadcast transmission}
		\textbf{Input:} Access probability vector $\mathbf{p}$, adjacency matrix $\mathbf{A}$, mixing matrix $\mathbf{W}$, initial parameters $\boldsymbol{x}$, number of iterations $T$ and step-size $\eta$.
		\begin{algorithmic}[1]
			\State $t\gets1$
			\While{$t \leq T$}
			\State $\mathbf{b}^{(t)} \sim Be(\mathbf{p})$  ~~\%Generate broadcast decisions by Bernoulli trials with probability $p$
			\State $\mathbf{T}^{(t)} \gets \tau(\mathbf{b}^{(t)}, \mathbf{A})$  ~~\%Generate transmission status matrix according to \eqref{transmission_dynamics}
			\State $\mathbf{W}^{(t)} \gets \mathbf{W}\odot\mathbf{T}^{(t)} $   ~~\%Obtain new weight matrix that includes link success/failure
			\State $\overline{\mathbf{W}}^{(t)} \gets w(\mathbf{W}^{(t)})$~~\%Apply biased compensation method according to \eqref{eq:w-compensated}
			\For{$i = \{1,2,...,N\}$}
			\State $\boldsymbol{x}_i^{(t+1)} \gets \sum_{j=1}^{N}\overline{W}_{i,j}^{(t)}\left[\boldsymbol{x}_j^{(t)}-\eta g(\boldsymbol{x}_i^{(t)})\right]$
			\EndFor
			\State $t \gets t + 1$
			\EndWhile
		\end{algorithmic}
		\vspace{-0.1cm}
	\end{algorithm}
	\subsection{Optimizing Access Probability}
	From the link success probability given in \eqref{eq:link-success}, we obtain the average number of successful links in the network as
	\begin{equation}\label{eq:random_policy_equal_prob}
		\mathbb{E}[N^{\text{suc}}] = \sum_{i\in \mathcal{V}, j \in\mathcal{N}(i)}^{}p_{i,j}^{\text{suc}}
		= p\sum_{i\in \mathcal{V}}d_i(1-p)^{d_i}.
	\end{equation}
	We refer to this as the expected throughput of the network.
	
	A common approach for fast convergence in distributed consensus or decentralized optimization is to minimize the second largest absolute eigenvalue of the mixing matrix (equivalently, the spectral radius of $\mathbf{W}-\frac{1}{N}\textbf{1}\textbf{1}^{\transp}$). In this work, our speculation is that with random-access-based broadcast, the average throughput serves as a natural approximation for measuring how well connected (in terms of successful information flow) a network is given the base topology. Therefore,
	we intend to find the throughput-optimal access probability. By taking the first-order derivative of $\mathbb{E}[N^{\text{suc}}]$ with respect to $p$ and setting it to 0, we obtain
	\begin{equation}\label{eq:d_random_policy_equal_prob}
		\begin{aligned}
			\frac{d\mathbb{E}[N^{\text{suc}}]}{dp} &=  \sum_{i\in \mathcal{V}}d_i(1-p)^{d_i} - p\sum_{i\in \mathcal{V}}d_i^2(1-p)^{d_i - 1}\\
			&=\sum_{i\in \mathcal{V}}d_i(1-p)^{d_i - 1}(1-p(1+d_i)) = 0.
		\end{aligned}
	\end{equation}
	Any $p\in(0,1)$ that satisfies the equality in \eqref{eq:d_random_policy_equal_prob} is a global optimal solution for maximizing $\mathbb{E}[N^{\text{suc}}]$.
	
	\begin{lemma}
		\label{lemma1}
		For a network modeled by a connected undirected graph with symmetric and circulant adjacency matrix $\mathbf{A}$ (e.g., ring and complete graphs), when all nodes access the channel and broadcast with the same probability $p$, we have 
		\begin{equation}
			\argmax\limits_{p} \mathbb{E}[N^{\text{suc}}]= \argmax\limits_{p}\rho \left(\mathbb{E}[\overline{\mathbf{W}}^{(t)}]-\textbf{1}\textbf{1}^{\transp}/N\right),
		\end{equation}
		where $\rho(A)$ means the spectral radius of a square matrix $A$.
	\end{lemma}

	\begin{proof}
		Due to space limit, the proof will be provided in a longer version of this paper.
		\vspace{-0.1cm}
	\end{proof}
	This lemma shows that choosing an access probability that maximizes the expected network throughput can be a good strategy for access control in decentralized learning.

	\section{Simulation Results}
	We created a network of $N=20$ nodes with two topologies: 1) Erdős–Rényi random graph; 2) ring graph, as illustrated in Fig.~\ref{fig:network}.
	Two simple learning tasks are considered: 1) regression; 2) classification.\footnote{More extensive simulation results using larger learning models and real data will be included in an extended version of this paper.}
	The regression task is to fit a horizontal line $\hat{y} = \theta$ for estimation of a bias with added Gaussian noise $ y = b + w, w \sim \mathcal{N}(0,\sigma^2)$. The classification task is to fit a linear model with softmax activation $\hat{\mathbf{y}} = \sigma(\boldsymbol{\uptheta}^T\mathbf{x})\in\mathbb{R}^4$ for classification of clusters. The cluster samples are generated by $\mathbf{x} = \mathbf{c}_j + \mathbf{w} \in\mathbb{R}^2, \mathbf{y} = j$ where $\mathbf{w}\sim \mathcal{N}(0, \Sigma)$ and $\mathbf{c}_j\sim \mathbf{U}(-1,1)$ for class index $j\in\{1,2,3,4\}$. The local objective functions are defined as the L2 and cross-entropy loss of predictions for the regression and classification tasks, respectively.

	\begin{figure}[t!]
		\centering
		\begin{subfigure}[b]{0.49\columnwidth}
			\centering
			\includegraphics[width=1.05\columnwidth]{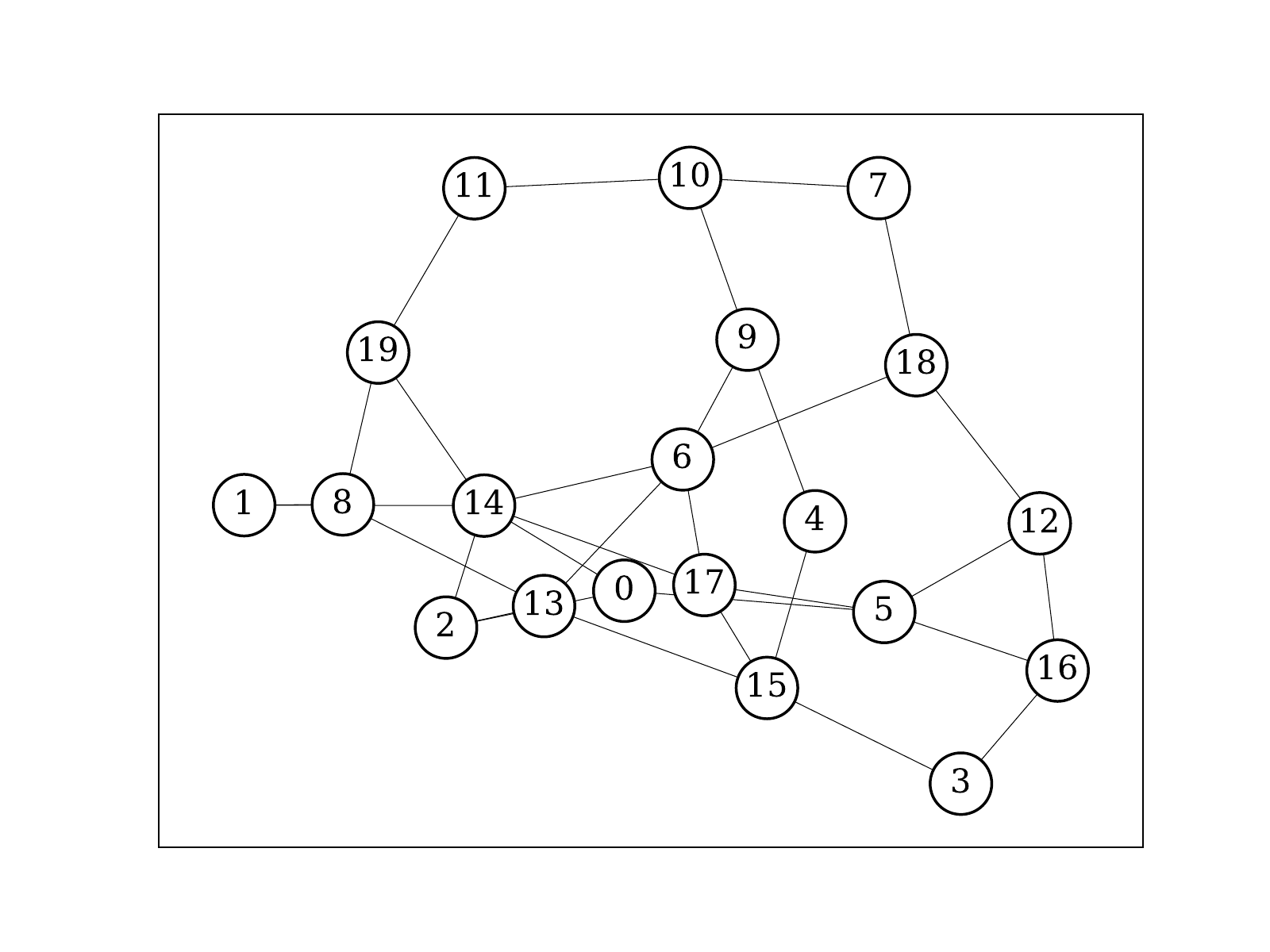}
			\caption{Erdős–Rényi}
		\end{subfigure}
		\begin{subfigure}[b]{0.49\columnwidth}
			\centering
			\includegraphics[width=1.05\columnwidth]{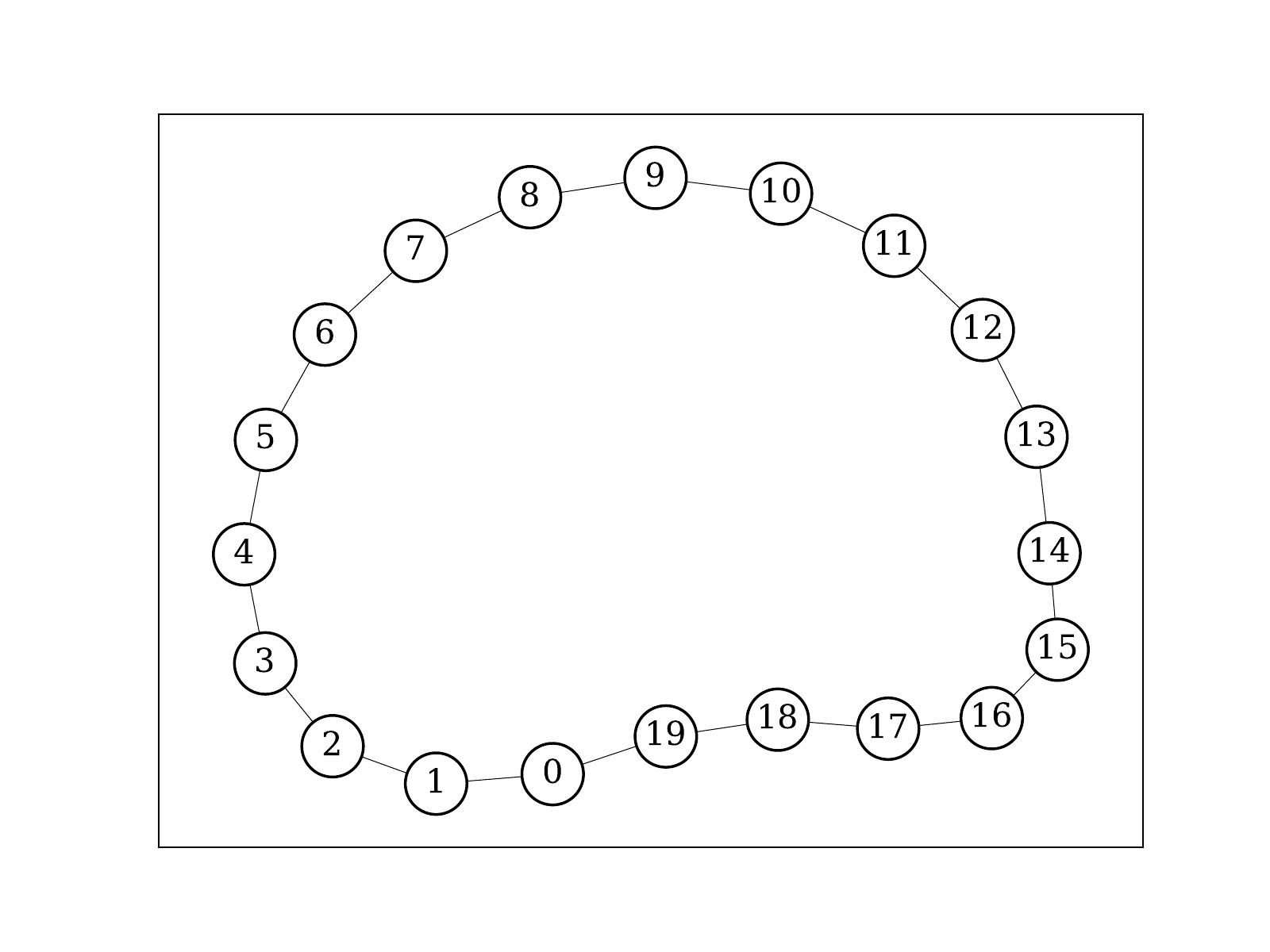}
			\caption{Ring }
		\end{subfigure}
		\caption{Two graph topologies for the simulations}
		\label{fig:network}
		\vspace{-0.15cm}
	\end{figure}
	
	\begin{figure}[t!]
		\centering
		\begin{subfigure}[b]{0.8\columnwidth}
			\centering
			\includegraphics[width=\columnwidth]{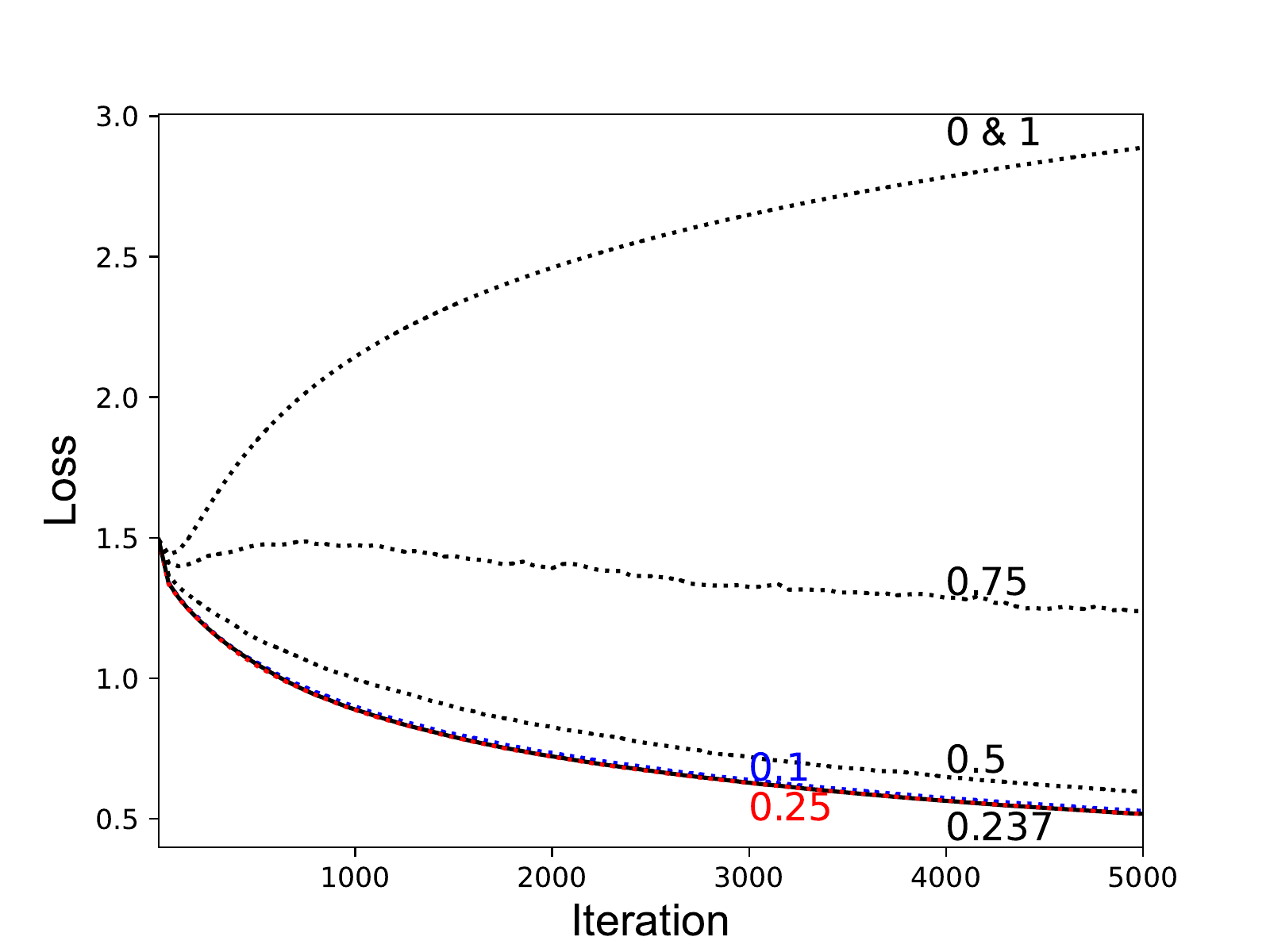}
			\caption{Classification task, loss over 5000 iterations}
			\label{fig:loss}
		\end{subfigure}
		\begin{subfigure}[b]{0.8\columnwidth}
			\centering
			\includegraphics[width=\columnwidth]{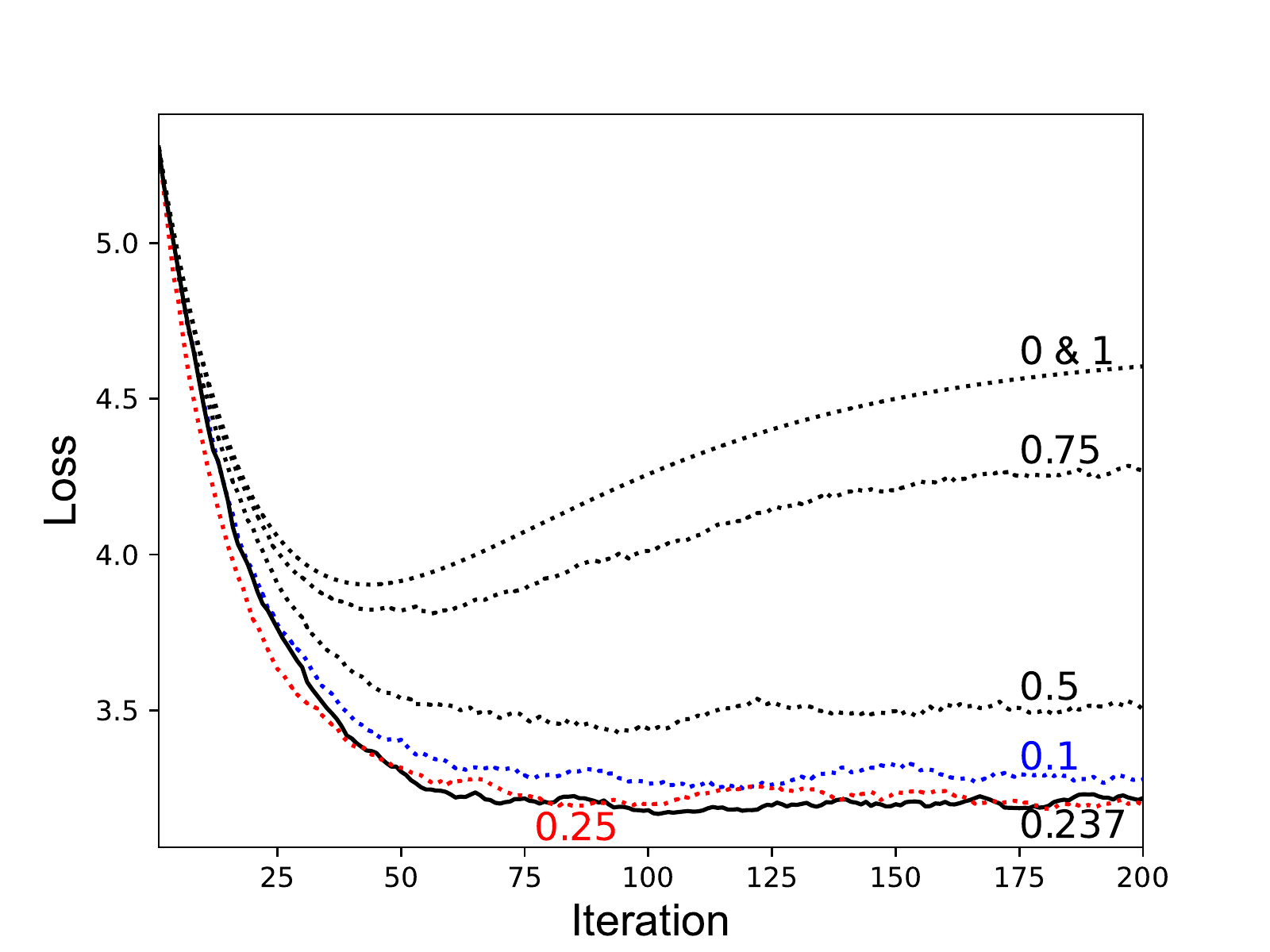}
			\caption{Regression task, loss over 200 iterations.}
			\label{fig:regression-loss}
		\end{subfigure}
		\caption{Erdős–Rényi random graph}
		\label{erdos-renyi-loss}
		\vspace{-0.15cm}
	\end{figure}
	
	\begin{figure}[t!]
		\centering		
		\begin{subfigure}[b]{0.8\columnwidth}
			\centering
			\includegraphics[width=\columnwidth]{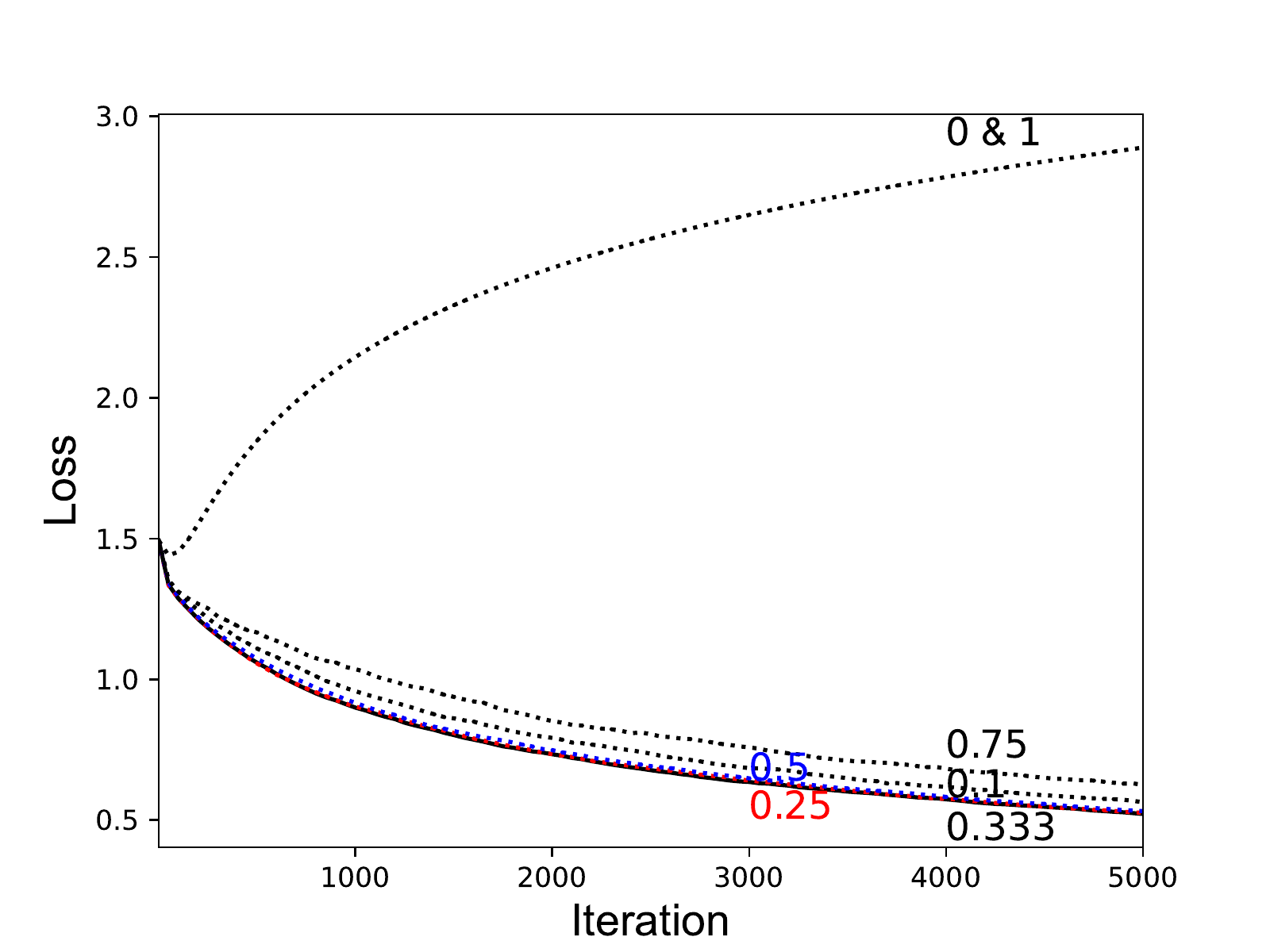}
			\caption{Classification task, loss over 5000 iterations}
			\label{fig:loss-ring}
		\end{subfigure}
		\begin{subfigure}[b]{0.8\columnwidth}
			\centering
			\includegraphics[width=\columnwidth]{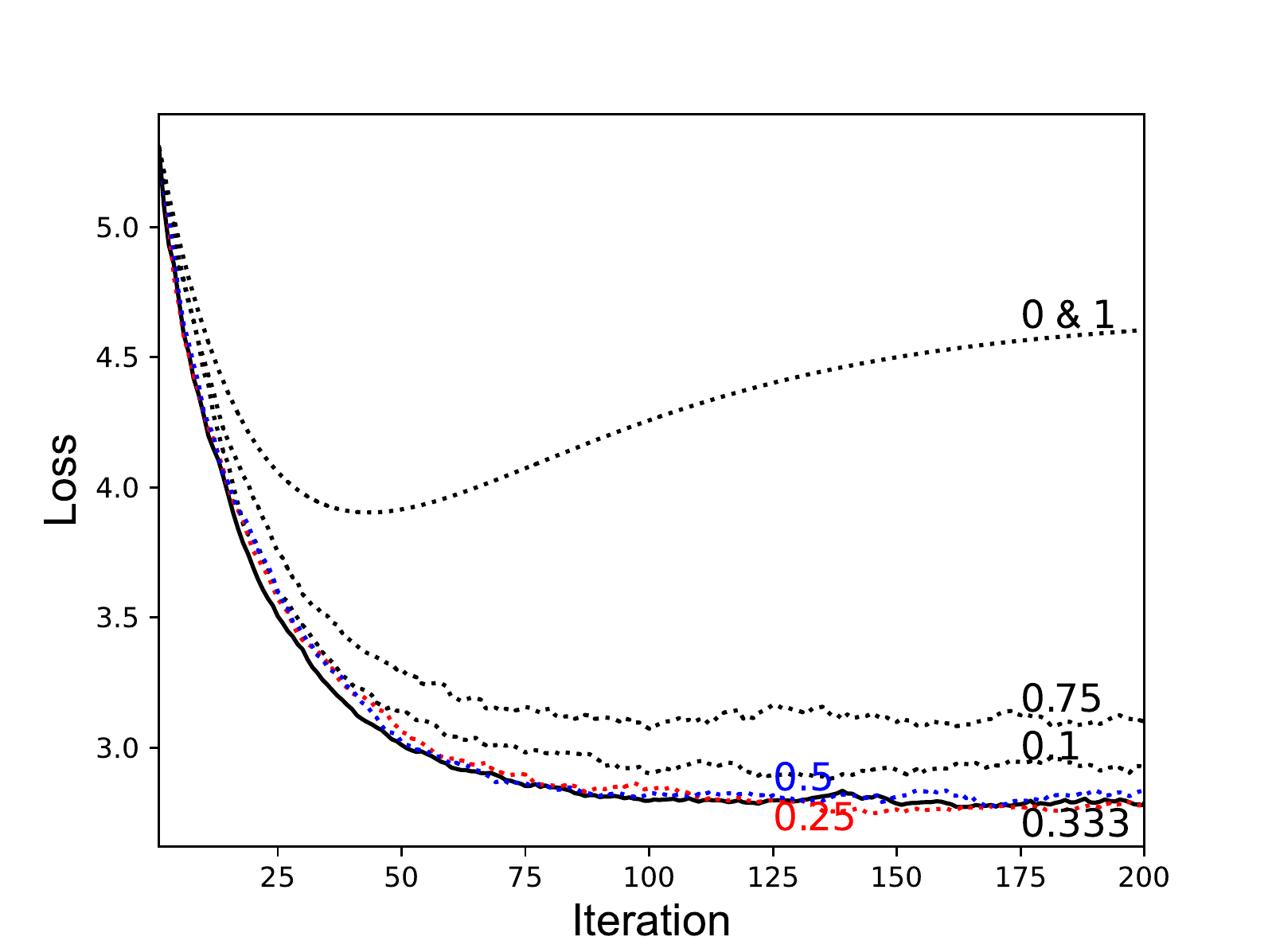}
			\caption{Regression task, loss over 200 iterations.}
			\label{fig:regression-loss-ring}
		\end{subfigure}
		\caption{Ring graph}
		\label{ring-loss}
	\end{figure}
	
	In the regression task, each local dataset $\mathcal{D}_i$ consists of $100$ independently and identically distributed (IID) samples and every node has its own bias value $b_i\sim \text{U}(-1,5)$. In the classification task, each $\mathcal{D}_i$ consists of $100$ IID samples from only one cluster class. Each class is represented by equally many local datasets. This setting creates a non-IID distribution of local datasets in both tasks. 
	The test datasets $\mathcal{D}_{\text{test}}$ in both tasks contain $100\cdot N$ additional samples for both data types, such that every bias value and cluster class have equally many samples. 
	
	The step-size $\eta$ is set as $0.01$, and local gradients $g_i$ are computed using a batch-size of $100$. 
	The performance of the trained model is evaluated by the average loss and accuracy of the local models on the test datasets $\mathcal{D}_{\text{test}}$. 
	\subsection{Effect of Access Probability on System Performance}
	First, in Fig.~\ref{erdos-renyi-loss}, we show the loss of the classification and regression tasks for Erdős–Rényi random graph, with different values of the access probability $p$. As we can see, $p=0$ and $p=1$ give the worst training performance, which is expected as in both cases there is no successful information exchange among the nodes. With non-IID training data, parallel training at different nodes without information fusion will generally lead to poor learning performance. Between $0$ and $1$ there is clearly an optimal value that gives the best result. In this example it corresponds to $p\approx 0.25$.
	In Fig.~\ref{ring-loss}, we present the same results for the ring graph, and we observe that the optimal access probability is $p\approx 0.333$.

	\subsection{Optimal Access Probability for Fast Convergence}
	In Fig.~\ref{fig:optimal-proba}, we show the relation between the optimal probability that maximizes the expected throughput defined in \eqref{eq:random_policy_equal_prob} and the one that minimizes the spectral radius of $\mathbb{E}[\overline{\mathbf{W}}^{(t)}]-\textbf{1}\textbf{1}^{\transp}/N$ (equivalently, the second largest absolute eigenvalue of $\mathbb{E}[\overline{\mathbf{W}}^{(t)}]$). For Erdős–Rényi random graph, these two values are $0.237$ and $0.25$, which are very close to each other. For the ring graph, both values are $0.333$, which can be further justified by our finding in Lemma~\ref{lemma1}.
	
	\begin{figure}[ht!]
		\centering
		\begin{subfigure}[b]{0.5\textwidth}
			\centering
			\includegraphics[width=\textwidth]{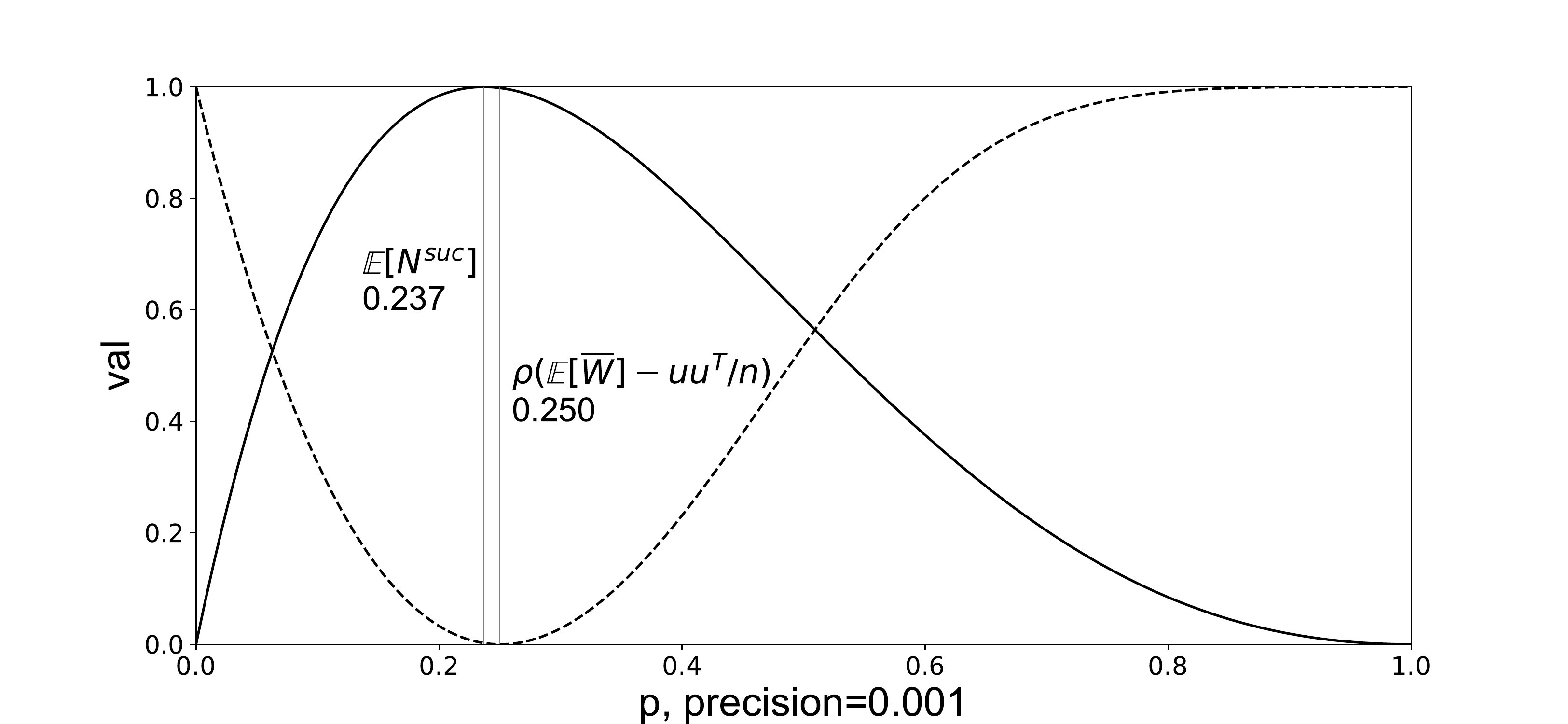}
			\caption{Erdős–Rényi random graph}
		\end{subfigure}
		\hfill
		\begin{subfigure}[b]{0.5\textwidth}
			\centering
			\includegraphics[width=\textwidth]{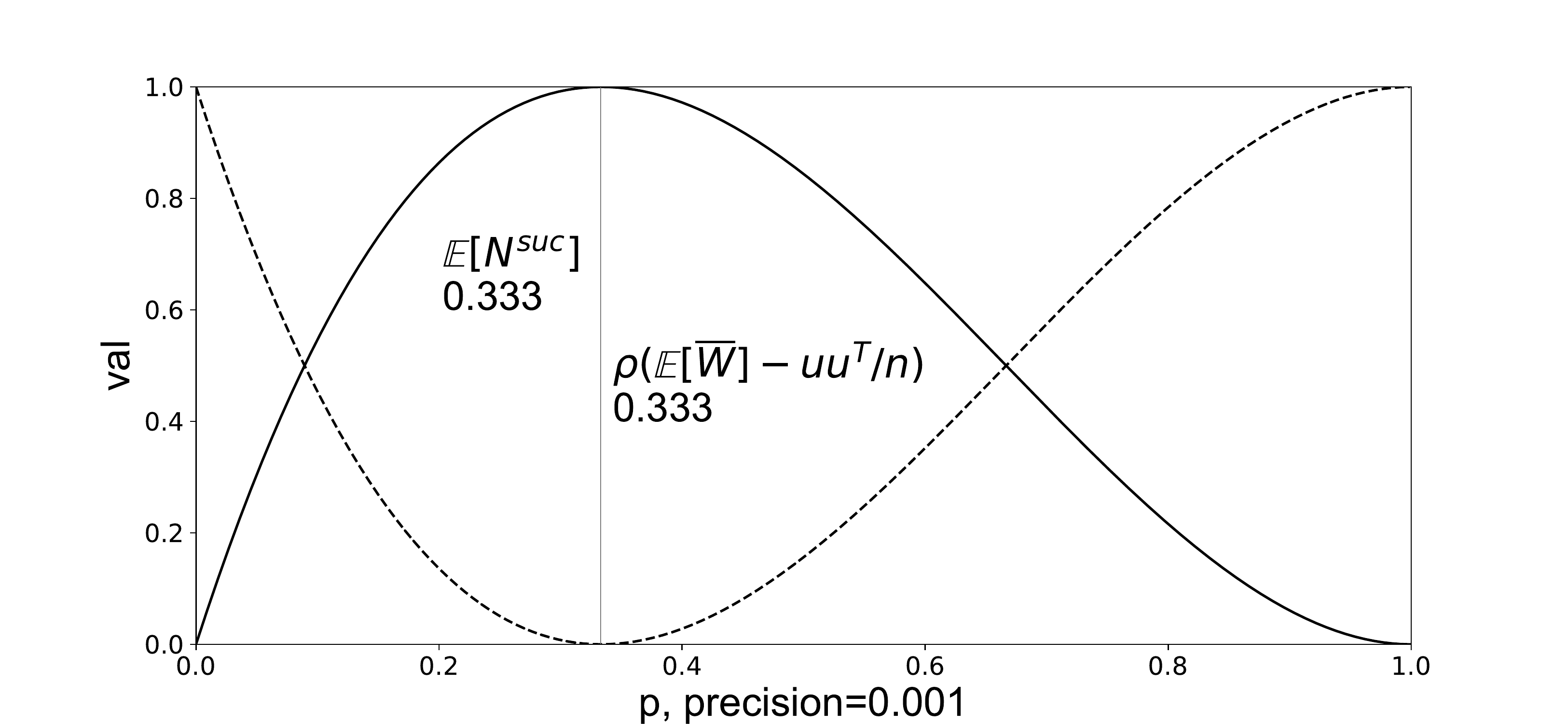}
			\caption{Ring graph}
	\end{subfigure}
	\caption{Comparison between the probability that maximizes $\mathbb{E}[N^{\text{suc}}]$ and the one that minimizes $\rho(\mathbb{E}[\overline{\mathbf{W}}^{(t)}]-\frac{1}{N}\textbf{1}\textbf{1}^{\transp})$. }
	\label{fig:optimal-proba}
	\vspace{-0.1cm}
\end{figure}

\subsection{Discussions}
Using the expected throughput (number of successful links) for measuring the level of information fusion in D-SGD is shown to be effective, but not always optimal.
A potential extension is to consider the importance of each node or link on the connectivity of the graph and the training data representation. This is particularly important for the non-IID data setting, as existing methods for accelerating D-SGD mostly focus on achieving faster convergence, but not on the accuracy of the converged model.
Introducing weights for different nodes or links could marginally improve the choice of the optimal access probability. 
\section{Conclusions}
This work aimed to investigate the effect of broadcast transmission and random access on the performance of decentralized learning over wireless networks. Based on a probabilistic random access scheme with success/collision model, we showed that fast convergence can be achieved by choosing an access probability that maximizes the expected number of successful links in the network. Furthermore, we provided theoretical proof for some special topologies, such as ring and complete graphs. As a future research direction, investigating random access with spatial separation in large-scale wireless networks would be an intriguing extension of this work. 


\end{document}